\documentclass[11pt]{amsart}
\usepackage{fullpage}
\usepackage{amssymb}
\usepackage{graphicx}
\usepackage{epstopdf}
\usepackage{url}
\usepackage{amsmath}
\usepackage{amsthm}
\usepackage{fullpage}
\usepackage{algorithm}
\usepackage{algorithmic}
\usepackage{mydefs}

\newlength{\tablength}
\newlength{\spacelength}
\newcommand{\tabstar}{\hspace*{\tablength}}
\newcommand{\spacestar}{\hspace*{\spacelength}}
\def\obeytabs{\catcode`\^^I=\active}
{\obeytabs\global\let^^I=\tabstar}
{\obeyspaces\global\let =\spacestar}

\newenvironment{display}{\begingroup\obeylines\obeyspaces\obeytabs}{\endgroup}
\newenvironment{prog}{\begin{display}\parskip0pt\sf}{\end{display}}

\begin{document}
\title[Wireless Capacity With Arbitrary Gain Matrix]{Wireless Capacity With Arbitrary Gain Matrix}

\author[M. Halld\'orsson]{Magn\'us M. Halld\'orsson}
\address[M. Halld\'orsson]{School of Computer Science\\
Reykjavik University\\
Reykjavik 101, Iceland}
\email{mmh@ru.is}

\author[P. Mitra]{Pradipta Mitra}
\address[P. Mitra]{School of Computer Science\\
Reykjavik University\\
Reykjavik 101, Iceland}
\email{ppmitra@gmail.com}

\begin{abstract}
Given a set of wireless links, a fundamental problem is to find the largest subset that can transmit simultaneously, within the SINR model of interference. Significant progress on this problem has been made in recent years. In this note, we study the problem in the setting where we are given a fixed set of arbitrary powers each sender must use, and an arbitrary gain matrix defining how signals fade. This variation of the problem
appears immune to most algorithmic approaches studied in the literature. Indeed it is very hard to approximate since 
it generalizes the max independent set problem.
Here, we propose a simple semi-definite programming approach to the problem that yields constant factor approximation, if the optimal solution is  strictly larger than half of the input size.
\end{abstract}

\maketitle

\section{Introduction}
We consider the fundamental problem of wireless network capacity.
Given is a set
$L = \{\ell_1, \ell_2, \ldots, \ell_n\}$ of links, where
each link $\ell_v$ represents a communication request from a sender
$s_v$ to a receiver $r_v$.  We are also given,
for every $\ell_v \in L$, a transmission power $P_v > 0$. 
The powers received from senders to receivers are defined by an $n \times n$ dimensional gain matrix $G$ with positive entries. Specifically, the signal received from $s_v$ at $r_w$ is $G_{wv} \cdot P_v$. Thus an instance in this model
can be described by the tuple $(L, P, G)$ where $P$ is the vector of the power assignments $P_v$ for all $\ell_v$.

 Simultaneously communicating links interfere with each other,
following the  physical model or ``SINR model" of interference.
Due to its higher fidelity to reality~\cite{GronkMibiHoc01,MaheshwariJD08,Moscibroda2006Protocol},
this model of interference has recently gained substantial  attention in the analysis of wireless networks.
In this model, a receiver $r_v$
successfully receives a message from a sender $s_v$ if and only if the
following condition holds:
\begin{equation}
 \frac{G_{vv} \cdot P_v }{\sum_{\ell_w \in S \setminus  \{\ell_v\}}
   G_{vw} \cdot P_w  + N} \ge \beta, 
 \label{eq:sinr}
\end{equation}
where $N$ is a universal constant denoting the ambient noise, $\beta \ge 1$ denotes the minimum
SINR (signal-to-interference-noise-ratio) required for a message to be successfully received,
and $S$ is the set of concurrently scheduled links in the same \emph{slot}. We say that a link $\ell_v$ is feasible in $S$ if Eqn. \ref{eq:sinr} is satisfied for $\ell_v$.
A set $S$ is feasible if each of its link is feasible.  

Note that what we described above is the abstract SINR model.
In the more commonly studied geometric SINR model, $G_{vw}$ is a polynomial function of $d(s_w, r_v)$,
where $d(x, y)$ is the distance between two points $x$ and $y$. Our results naturally apply to that model
as well. Given that the geometric SINR model does not capture obstacles, reflections and other real life distortions, 
it is interesting to see what can be proven in the abstract model.

Our setting where the powers are given as part of the input is often called the \emph{fixed} power case,
as opposed to the \emph{power control} case where the algorithm can choose the power assignment.
So far, research on fixed power
has focused on \emph{oblivious} power assignments, where the power
of a link is a (usually simple) function of the length of the link \cite{HW09,FKRV09,KV10,SODA11}.
Recently, a constant factor approximation algorithm to find the capacity in the power control case 
has also been achieved \cite{KesselheimSoda11}. Unfortunately, none 
of these techniques appear to extend to the case of arbitrary fixed powers (for either arbitrary or geometric gain matrices). Yet, the problem of arbitrary fixed powers
is not only natural, but has practical relevance, as commercial hardware often do not 
have the capacity of choosing precise powers to implement either an arbitrary assignment \emph{\`{a} la} \cite{KesselheimSoda11}, or to implement many of the oblivious power assignments found in literature.

In this paper, we prove the following theorem.
\begin{theorem}
Assume $(L, P, G)$ is an instance of the capacity problem in the abstract SINR model, such that $|OPT| > \frac{1}{2}(1 + \epsilon) |L|$ for some $\epsilon > 0$, where $OPT$ is the maximum feasible subset of $L$ using $P$. Then there is a polynomial time randomized algorithm to find a feasible set of size $\Omega(\epsilon |L|)$, with probability $1 - o(1)$.
\label{mainth1}
\end{theorem}

We do this by means of a semi-definite programming relaxation, which we show how to successfully round if the condition 
 $|OPT| > \frac{1}{2}(1 + \epsilon) |L|$  holds. In addition, we discuss numerical experiments we have performed. These experiments show that the algorithm appears to work quite well on random instances, even better than the guarantees of
 Thm. \ref{mainth1}.
 
 Semi-definite programming has been a staple in designing approximation algorithms for NP-hard problems ever
 since the seminal work of Goemans and Williams on the Max-CUT problem \cite{Goemans:1995:IAA:227683.227684}. 
 It is interesting to note that the discrete ``classical" problems closest to wireless capacity, namely the independent set
 problem and the graph coloring problem, have been  fruitfully studied using semi-definite programming \cite{Halperin:2002:CK-:606216.606221,Karger:1998:AGC:274787.274791}. The vertex cover problem, also relevant via its connection
 to the independent set problem,  also has
SDP-based approximation algorithms 
\cite{Halperin:2000:IAA:338219.338269,Karakostas:2009:BAR:1597036.1597045}. Given this background, one may expect some of the techniques to easily carry over to the capacity
 problem. Yet that does not appear to be the case, at least not in a straightforward manner. A study of the aforementioned
 papers reveal that the discreteness of the problem plays an important role in the bounds. For example, in 
\cite{Karger:1998:AGC:274787.274791}, the analysis proceeds by bounding the probability of vectors representing edges not being cut by a random hyperplane. Given the additive nature of the SINR model, it is not obvious
how to extend that analysis to this case. There  have also been a number of results for these problems on hypergraphs \cite{Krivelevich:2001:ACM:365411.365469,DBLP:conf/focs/Chlamtac07,DBLP:conf/approx/ChlamtacS08}. Though hypergraphs appear to be closer in spirit to the additive wireless model, they
are still different, because the effect of each node on 
any other node doesn't change in the SINR model (as opposed to in a hypergraph, where it can be different based on which edge they are in). Thus, the (sophisticated) methods on hypergraphs 
do not appear to  translate immediately to the SINR model either. Our SDP relaxation and rounding algorithms are
quite simple in contrast to some of the previously mentioned work. Whether or not advanced techniques can be extended to the SINR model remains to be seen.

\subsection{Related Work.}

Moscibroda and Wattenhofer \cite{MoWa06} were the first to
study of the \emph{scheduling complexity} of arbitrary set
of wireless links. 
Early work on approximation algorithms
 produced approximation
factors that grew with structural properties of the network \cite{moscibroda06b,MoscibrodaOW07,chafekar07}.

The first constant factor approximation algorithm was obtained for
capacity problem for uniform power in \cite{GHWW09} (see also
\cite{HW09}) in $\mathbf{R^2}$ with $\alpha > 2$.
Fangh\"anel, Kesselheim and V\"ocking \cite{FKV09} gave an algorithm
that uses at most $O(OPT + \log^2 n)$ slots for the scheduling problem
with \emph{linear} power assignment $P_v = d(s_v, r_v)^\alpha$,
that holds in general distance metrics.

Kesselheim obtained 
a $O(1)$-approximation algorithm
for the capacity problem with power control for doubling metrics \cite{KesselheimSoda11}. Around the same time,
the first constant factor algorithm for all sub-linear, length monotone power assignments
was achieved on general metrics \cite{SODA11}. Other recent studies in the SINR model
include work on topological maps \cite{stoc_topology11}, distributed algorithms for scheduling \cite{icalp11},
distributed power control \cite{damsicalp11} and auction based spectrum allocation \cite{hoeferspaa11}.



\section{SDP-based algorithm.}

First, some notation. Vectors are denoted by $\vec{x}, \vec{s_w}$ etc. The standard $2$-norm of the vector $\vec{x}$ is $\|\vec{x}\|$. The $i^{th}$
entry of $\vec{x}$ is $\vec{x}(i)$.
The inner product of vectors $\vec{x}$ and $\vec{y}$ is denoted $(\vec{x} \cdot \vec{y})$. Define $g_{vv} = P_v G_{vv} - \beta N$ and $g_{vw} = P_w G_{vw}$ for $v \neq w$. 
Note that we can assume without loss of generality that $g_{vv} \geq 0, \forall v$.
Let $OPT$ be a feasible subset of $L$
of maximum size. Note that $n = |L|$.

Consider the following program.

\begin{eqnarray*}
& & \max \sum_v (\vec{s_v} \cdot \vec{s}), \text{ subject to}\\
& & (\vec{s_v} \cdot \vec{s}) g_{vv} \geq \beta \left(\sum_{w \neq v} (\vec{s_v} \cdot \vec{s_w}) g_{vw}\right), \forall v\\
& & (\vec{s_v} \cdot \vec{s}) \geq 0, \forall v\\
& & (\vec{s_v} \cdot \vec{s_w}) \geq 0, \forall v,w \\
& & (\vec{s_v} \cdot \vec{s_w}) \geq (\vec{s_v} \cdot \vec{s}) + (\vec{s_w} \cdot \vec{s}) - 1, \forall v,w \\
& & \|\vec{s_v}\|^2 = 1, \forall v  \text{ and } \|\vec{s}\|^2 = 1 \ .
\end{eqnarray*}
where $\vec{s_v}, \vec{s} \in \mathbb{R}^{n+1}$. Each link $\ell_v$ has a vector variable $\vec{s_v}$ associated
with it. The dot product of $\vec{s_v}$ with a vector $\vec{s}$ denotes  the (fractional) extent to which $\ell_v$ is
selected in the solution. 

Since the objective function and constraints are all linear functions of vector inner products, this
problem is a SDP. Thus the program can be solved up to an additive error of $\varepsilon > 0$ in time that is polynomial in 
$n$ and $\log \varepsilon$ \cite{Vandenberghe94semidefiniteprogramming}. Since $\varepsilon$ can be made small enough to not matter, we will simply
assume that the problem can be solved exactly.

We can rotate the vectors to fix $\vec{s} = \{1, 0 \ldots 0\}$, thus the above program is equivalent to:

\begin{eqnarray}
& & \max \sum_v \vec{s_v}(1), s.t. \nonumber \\
& & \vec{s_v}(1) g_{vv} \geq \beta\left(\sum_{w \neq v} (\vec{s_v} \cdot \vec{s_w}) g_{vw}\right), \forall v \label{sdp:eqSINR}\\
& & \vec{s_v}(1) \geq 0, \forall v \label{sdp:o1}\\
& & (\vec{s_v} \cdot \vec{s_w}) \geq 0, \forall v,w \label{sdp:o2} \\
& & (\vec{s_v} \cdot \vec{s_w}) \geq \vec{s_v}(1) + \vec{s_w}(1) - 1, \forall v,w \label{sdp:eqSp}\\
& & \|\vec{s_v}\|^2 = 1, \forall v  \label{sdp:o3}\ .
\end{eqnarray}

Let us verify that this program is a relaxation of the maximum capacity problem. 
\begin{lemma}
The SDP is a relaxation of the original problem.
\end{lemma}
\begin{proof}
Consider any optimal solution $OPT$ to the capacity problem.
For all $\ell_v \in OPT$, set $\vec{s_v} = \vec{s} =  \{1, 0, 0, 0 \ldots 0\}$. 
If $\ell_v \in L \setminus OPT$ 
set 
\begin{equation*} 
\vec{s_v}(i) = \left\{ 
\begin{array}{rl}
1 & \text{if } i = v + 1\\
0 & \text{otherwise}\\
\end{array} \right.
\end{equation*}
In other words, we make sure that each unselected link chooses a different position
for the single $1$ in the vector. 

Given these assignments, Equations \ref{sdp:o1}, \ref{sdp:o2} and \ref{sdp:o3} can easily seen to hold.

To show that Eqn. \ref{sdp:eqSINR} is satisfied, first assume $\ell_v \in OPT$. The following observation is immediate:
\begin{claim}
If $\ell_v, \ell_w \in OPT$ then $\vec{s_v}(1) = \vec{s_w}(1) = (\vec{s_v} \cdot \vec{s_w}) = 1$. If $\ell_v \in L \setminus OPT$
then $\vec{s_v}(1) =0$ and $(\vec{s_v} \cdot \vec{s_w}) = 0$ for any $\ell_w \neq \ell_v$.
\end{claim}

Since $\ell_v \in OPT$,
\begin{equation*}
\vec{s_v}(1) g_{vv} = g_{vv}
\end{equation*}
And,
\begin{eqnarray*}
\lefteqn{\beta\left(\sum_{w \neq v} (\vec{s_v} \cdot \vec{s_w}) g_{vw}\right) } \\
& = & \beta\left(\sum_{w \in OPT \setminus \{v\}} (\vec{s_v} \cdot \vec{s_w}) g_{vw}\right)  + \beta\left(\sum_{w \in L \setminus (OPT \cup \{v\})} (\vec{s_v} \cdot \vec{s_w}) g_{vw}\right) \\
& =&  \beta\left(\sum_{w \in OPT \setminus \{v\}} g_{vw} \right) 
\end{eqnarray*}
where the second equality follows from the claim above.

Now, since $\ell_v \in OPT$, $g_{vv} \geq \beta\left(\sum_{w \in OPT \setminus \{v\}} g_{vw}\right)$ (by Eqn \ref{eq:sinr}). Thus, the above two equations show that Eqn. \ref{sdp:eqSINR} is satisfied when $\ell_v \in OPT$. The
case where $\ell_v \not \in OPT$ is similar.

For Eqn. \ref{sdp:eqSp}, the following observations suffice:
\begin{itemize}
\item If $\ell_v, \ell_w \in OPT$, $(\vec{s_v} \cdot \vec{s_w}) = 1 = \vec{s_v}(1) + \vec{s_w}(1) - 1$  
\item If $\ell_v, \ell_w \not \in OPT$, they have $1$s in different positions and $(\vec{s_v} \cdot \vec{s_w}) = 0 \geq 0 +0 - 1$  
\item If $\ell_v \in OPT, \ell_w \not \in OPT$, they have 1s in different positions and $(\vec{s_v} \cdot \vec{s_w}) = 0 = 1 +0 - 1$  
\end{itemize}
\qed
\end{proof}


Now we present our algorithm and the proof of Thm. \ref{mainth1}. We need two related definitions.
Let $\delta_v = \max\{\vec{s_v}(1) - \frac{1}{2}, 0\}$ for all $\ell_v \in L$. Further, define $L^+ = \{\ell_v \in L: \delta_v > 0\}$. The algorithm is as follows.

\begin{algorithm}                      
\caption{Capacity1}          
\label{alg1}                           
\begin{algorithmic}[1]                    
     \STATE Solve the SDP
     \STATE Select each link $\ell_v \in L^+$ with probability $\frac{\delta_v}{2}$ in to a set $R$ 
     \STATE Output $\{\ell_v \in R : \ell_v \text{ is feasible in } R\}$
\end{algorithmic}
\end{algorithm}

\begin{lemma}
\label{sumdelta}
If $|OPT| \geq (1 + \epsilon) n/2$, then
$\sum_{\ell_v \in L^+} \delta_v \geq \frac{n \epsilon}{2}$.
\label{lemma1}
\end{lemma}
\begin{proof}
Since $|OPT| \geq (1 + \epsilon) n/2$, it follows that $\sum_v \vec{s_v}(1) \geq (1 + \epsilon) n/2$ (since the SDP is a 
relaxation of the original problem). Now by definition of $\delta_v$, $\delta_v + \frac{1}{2} \geq \vec{s_v}(1)$.
Thus, 
\begin{eqnarray*}
& & \sum_{\ell_v \in L} \left(\frac{1}{2} + \delta_v\right) \geq (1 + \epsilon) n/2 \\
& \Rightarrow & \sum_{\ell_v \in L}  \delta_v \geq (1 + \epsilon) n/2 - |L|/2 = (1 + \epsilon) n/2 - n/2 = \frac{\epsilon n}{2}
\end{eqnarray*}

Observing that $\delta_v = 0$ for $\ell_v \not \in L^+$ completes the proof. \qed
\end{proof}

\noindent We can now prove the main Theorem.

\begin{proof}{of Thm. \ref{mainth1}}
%
%
%


 Assume  that the random binary variable $X_v$ describes whether or not $\ell_v \in L^+$ is chosen into $R$. We observe that
 $\Ex(X_v) = \frac{\delta_v}{2}$, according to the algorithm.
 
 Then for any $\ell_v$,
 \begin{eqnarray}
&&  \Ex\left(\beta\left(\sum_{w \in R \setminus \{v\}} g_{vw}\right)\right) = \Ex\left(\beta\left(\sum_{w \in L^+ \setminus \{v\}} g_{vw} X_w \right)\right) \nonumber \\
& = & \beta\left(\sum_{w \in L^+ \setminus \{v\}} g_{vw} \Ex(X_w) \right) = \beta\left(\sum_{w \in L^+ \setminus \{v\}} g_{vw} \frac{\delta_w}{2} \right) \nonumber \\ 
& \Rightarrow & \Ex\left(\beta\left(\sum_{w \in R \setminus \{v\}} g_{vw}\right)\right) = \frac{1}{2}\beta\left(\sum_{w \in L^+ \setminus \{v\}} g_{vw} \delta_w \right) \label{infeasibprob1}
 \end{eqnarray}
 
Now, by Eqn. \ref{sdp:eqSINR},
$$\vec{s_v}(1) g_{vv} \geq \beta\left(\sum_{w \neq v} (\vec{s_v} \cdot \vec{s_w}) g_{vw}\right), \forall v \in L^+$$
Since $\vec{s_v}(1) \geq \frac{1}{2}$ for $v \in L^+$ and $(\vec{s_v} \cdot \vec{s_w}) g_{vw}$ is always non-negative, we get for $\ell_v \in L^+$,
\begin{eqnarray}
g_{vv} & \geq & \beta\left(\sum_{w \in L^+ \setminus \{v\}} (\vec{s_v} \cdot \vec{s_w}) g_{vw}\right) \nonumber \\
& \geq & \beta\left(\sum_{w \in L^+ \setminus \{v\}} (\vec{s_v}(1) + \vec{s_w}(1) - 1) g_{vw}\right)
= \beta\left(\sum_{w \in L^+ \setminus \{v\}} (\delta_v+ \delta_w) g_{vw}\right) \nonumber \\
& \geq & \beta\left(\sum_{w \in L^+ \setminus \{v\}} \delta_w g_{vw}\right) \label{infeasibprob2}
\end{eqnarray}
where the second inequality follows from Eqn. \ref{sdp:eqSp}, and the first equality follows from observing that $\delta_v = \vec{s_v}(1) - \frac{1}{2}$ for $\ell_v \in L^+$.
 
 Then, for $\ell_v \in L^+$,
\begin{eqnarray}
 \Pro(\ell_v \text{ is infeasible in $R$}) && =  \Pro\left(\beta\left(\sum_{w \in R \setminus \{v\}} g_{vw}\right) > g_{vv}\right) \nonumber \\
&& \leq \frac{\Ex(\beta(\sum_{w \in R \setminus \{v\}}  g_{vw}))}{g_{vv}} \leq \frac{1}{2}
\end{eqnarray}
The first equality is the definition of infeasiblity. The first inequality is Markov's inequality. The last inequality
follows from Equations \ref{infeasibprob1} and \ref{infeasibprob2}.

Now the expected size of the output is

\begin{eqnarray*}
\Ex\left(|\{\ell_v \in R : \ell_v \text{ is feasible in } R\}|\right)
&& = \sum_{\ell_v \in L^+} \Pro(\ell_v \in R \text{ and } \ell_v \text{ is feasible in } R) \\
&& = \sum_{\ell_v \in L^+} \Pro(\ell_v \text{ is feasible in } R) \Pro(\ell_v \in R) \\
&& \geq  \sum_{\ell_v \in L^+} \frac{1}{2} \frac{\delta_v}{2} 
  = \frac{1}{4}\sum_{\ell_v \in L^+} \delta_v \geq \frac{n \epsilon}{8} \label{eqn:infeasibbound}
\end{eqnarray*}
The second equality follows from the independence of the events concerned. 
The first inequality follows from Eqn. \ref{eqn:infeasibbound}.
The last
inequality follows from Lemma \ref{lemma1}.
Thus the expected size of the feasible output is $\Omega(\epsilon n)$.
It is not difficult to boost the probability of getting a $\Omega(\epsilon n)$ size subset to complete
the proof of the theorem. \qed
\end{proof}

\section{Numerical Experiments}
We ran  simulations to test how well the algorithm does in practice. We used \texttt{CVX}, a package for specifying and solving convex programs using MATLAB \cite{cvx}. We ran it on version 7.8 of MATLAB running on a Macbook with a 2 GHz Intel Core 2 Duo Processor and 2 GB of RAM. 

We generated a number of problem instances where $n = 61$ and $|OPT| = 21, 26, 31, 36$ and $41$. The instances
were generated as follows. To generate the feasible subset a large random instance $M$ of links on the 2d plane was generated. Each sender $s_v = (s_v(x), s_v(y))$ is a random point in a $450 \times 450$ box. The receiver $r_v$ is defined by $(s_v(x) + \text{random}_v(x), s_v(y) + \text{random}_v(y))$ where $\text{random}_v(x)$ and $\text{random}_v(y)$ are sampled uniformly at random from $[-20, 20]$. We generated  corresponding gain matrices using the geometric SINR model setting $\alpha = 2.5$ (thus $G_{vw} = \frac{1}{\|s_w - r_v\|^{\alpha}}$). We used both uniform ($P_v$ is a constant) and mean power assignments ($P_v = \|s_v - r_v\|^{\alpha/2}$) to generate the gain matrix. We set the noise $N = 0$ throughout the experiments.

To generate the input instance $G$ (which is a $n \times n$ matrix), we combined a subset of $M$ with random entries.
More specifically, first we retrieved a random feasible subset $R$ of $M$ (found greedily). This defined a
$R \times R$ submatrix of $G$. The remaining entries were chosen iid randomly from $[0, \kappa]$, where
$\kappa$ was chosen large enough so that the remaining $n - |R|$ links would not contain a large feasible subset, thus $R$ would be $OPT$ for the instance.

Though computationally slow (for $n = 60$ the SDP took a few minutes to be solved), the algorithm performed extremely well. Indeed, it took some time to come up with instances where the algorithm didn't have a perfectly integral solution. If the random entries of $G$ corresponding to $L \setminus OPT$ were too large (corresponding to a large $\kappa$, meaning that $L \setminus OPT$ contained only very small subsets that were feasible) or if $OPT$ was too \emph{loosely} feasible (ie, Eqn. \ref{sdp:eqSINR} was far from being tight for most of the links), the algorithm did exceedingly well.

\begin{figure}
\begin{center}
\includegraphics[height=2.7in]{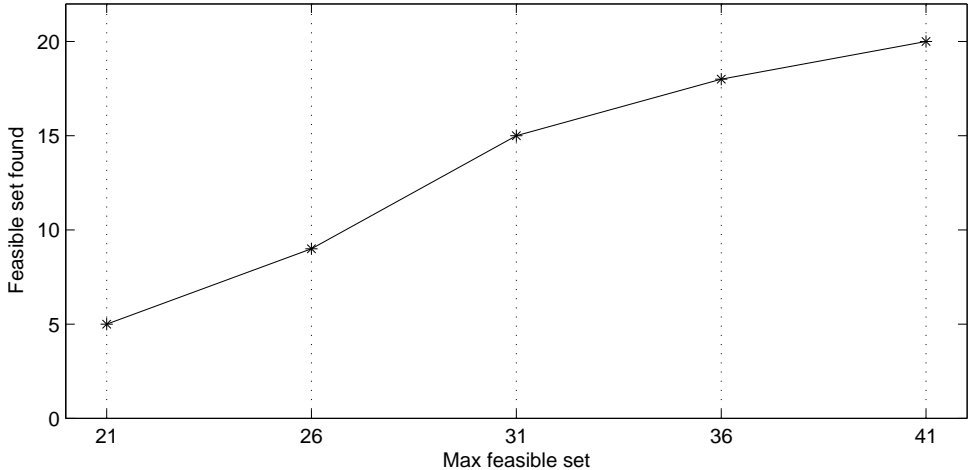}
\caption{$OPT$ vs the average size of the set found by the SDP algorithm. In each case $n = 61$.} \label{fig:uniform}
\end{center}
\end{figure}

\begin{figure}
\begin{center}
\includegraphics[height=2.7in]{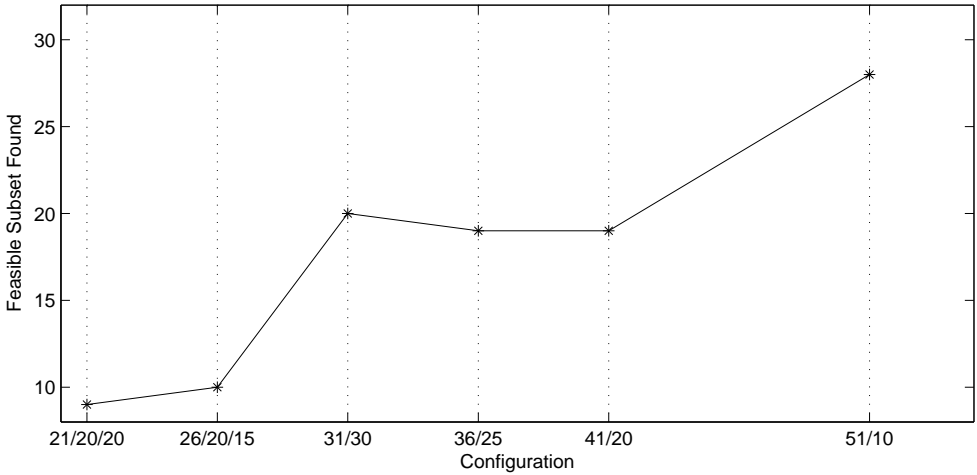}
\caption{$OPT$ vs the average size of the set found by the SDP algorithm. In each case $n = 61$ and the links in $OPT$ were generated using mean power. The labels in the $x$-axis describe the configuration of the instances. Thus, in the first case, the instance is an union of $3$ feasible sets of size 21, 20 and 20, respectively, where the latter two are copies of subsets of the first one.} \label{fig:mean}
\end{center}
\end{figure}

Even after trying to make the problem more difficult, the algorithm did quite well, only degrading when $OPT < n/2$, for which we claim no theoretical guarantee anyway, though even in these cases the output was not unsatisfactory. Indeed,
in all these cases, using the simple filtering $(\vec{s_v} \cdot \vec{s}) > 0.51$ identified $OPT$ almost exactly. Our sampling algorithm, by design cannot achieve better than a factor $2$ approximation in general, and that is almost what we achieved in all cases, as illustrated in Figure \ref{fig:uniform} for uniform power (the results for mean power were essentially identical).

As we have mentioned, in the above experiments, the algorithm sharply identified $OPT$. To create more
ambiguous instances, we also tried the following. In this setting, we took a feasible set, and added copies of subsets of
it. Thus the instance would be of the form $L_1 \cup L_2$ or $L_1 \cup L_2 \cup L_3$ where $L_1$ is feasible, and
$L_2, L_3$ are copies of subsets of $L_1$. One expects the solution to be more ``spread out" in this case, and that is exactly what we found.
The algorithm still performed rather well, even below theoretically guaranteed levels, though the behavior is somewhat different. Figure \ref{fig:mean} demonstrates the case for mean power.

\section{Conclusion}
We have shown how to use semi-definite programming to approximate the wireless capacity problem in cases where the capacity is known to be large. It is an interesting question whether or not these results can be further improved, potentially using the power of geometric SINR model. Questions about the integrality gap and hardness of the problem (apart from what is known via the fact that the problem generalizes max independent set) also deserve attention. Though we have performed some preliminary numerical experiments, the efficacy of this method both in terms of accuracy and computational efficiency also is an interesting avenue of further investigation.

\bibliographystyle{plain}
\bibliography{references}

\end{document}